\documentclass{sig-alternate-05-2015}

\usepackage{comment, amsmath, alltt, amssymb, xspace, times, epsfig, algorithm, url, subfigure, xcolor, color, bm, amsbsy}

\usepackage[noend]{algpseudocode}

\usepackage{verbatim}
\usepackage{multirow}
\usepackage{color}
\usepackage{multicol}

\newtheorem{theorem}{Theorem}
\newtheorem{lemma}[theorem]{Lemma}

\newtheorem{remark}{Remark}
\newtheorem{definition}{Definition}

\newcommand{\tabincell}[2]{\begin{tabular}{@{}#1@{}}#2\end{tabular}}

\newcommand{\sminus}{\ensuremath{\!-\!}}
\newcommand{\splus}{\ensuremath{\!+\!}}
\newcommand{\IPOSI}{\ensuremath{\mathbf{I}_\top}}
\newcommand{\INEGA}{\ensuremath{\mathbf{I}_\bot}}

\newcommand{\GPTT}{\ensuremath{\mathsf{GPTT}}}

\newcommand{\mypara}[1]{\vspace*{0.06in}\noindent\textbf{#1} \xspace}

\newcommand{\Lap}[1]{\ensuremath{\mathsf{Lap}\left(#1\right)}\xspace}

\renewcommand{\Pr}[1]{\ensuremath{\mathsf{Pr}\!\left[#1\right]}\xspace}
\newcommand{\myexp}[1]{\ensuremath{e^{#1}}\xspace}
\newcommand{\MM}{\ensuremath{\mathcal{M}}\xspace}

\renewcommand{\AA}{\ensuremath{\mathcal{A}}\xspace}

\newcommand{\nega}{\bot}
\newcommand{\posi}{\top}

\renewcommand{\vec}[1]{\mbox{\boldmath$#1$}}

\newenvironment{ecompact}{
	\vspace{-0.05in}
	\begin{enumerate}
		\setlength{\itemsep}{1pt}
		\setlength{\parindent}{0pt}
		\setlength{\parskip}{0pt}
		\setlength{\parsep}{0pt}
	}  {\end{enumerate}
	\vspace{-0.1in}
}

\DeclareMathAlphabet      {\mathbfit}{OML}{cmm}{b}{it}

\definecolor{forest}{RGB}{34, 139, 34}

\allowdisplaybreaks

\begin{document}
\sloppypar

\title{Understanding the Sparse Vector Technique for \\Differential Privacy}
\author{Min Lyu, Dong Su, Ninghui Li}

\author{%
	{Min Lyu$^{~\#}$, Dong Su$^{~\star}$, Ninghui Li$^{~\star}$}%
	\vspace{1.6mm}\\
	\fontsize{10}{10}\selectfont\itshape
	$^{\#}$\,
    University of Science and Technology of China \hspace{3cm}$^{\star}$\,Purdue University\\
	\fontsize{9}{9}\selectfont\ttfamily\upshape
	\hspace{2cm}lvmin05@ustc.edu.cn \hspace{3.5cm}\{su17, ninghui\}@cs.purdue.edu%
}


\maketitle

\newcommand{\update}[3]{#2\xspace}

\newcommand{\refertoappendix}[2]{#1}

\begin{abstract}
The Sparse Vector Technique (SVT) is a fundamental technique for satisfying differential privacy and  has the unique quality that one can output some query answers without apparently paying any privacy cost.  SVT has been used in both the interactive setting, where one tries to answer a sequence of queries that are not known ahead of the time, and in the non-interactive setting, where all queries are known.
Because of the potential savings on privacy budget, many variants for SVT have been proposed and employed in privacy-preserving data mining and publishing.  However, most variants of SVT are actually not private.  In this paper, we analyze these errors and identify the misunderstandings that likely contribute to them.  We also propose a new version of SVT that provides better utility, and introduce an effective technique to improve the performance of SVT.  These enhancements can be applied to improve utility in the interactive setting. 
Through both analytical and experimental comparisons, we show that, in the non-interactive setting (but not the interactive setting), the SVT technique is unnecessary, as it can be replaced by the Exponential Mechanism (EM) with better accuracy.

\end{abstract}

\pagestyle{plain}
\section{Introduction}
Differential privacy (DP) is increasingly being considered the privacy notion of choice for privacy-preserving data analysis and publishing in the research literature.
In this paper we study the Sparse Vector Technique (SVT), a basic technique for satisfying DP, which was first proposed by Dwork et al.~\cite{DNR+09} and later refined in~\cite{RR10} and~\cite{HR10}, and used in~\cite{Gupta2012,LC14,SCM14,CXZX15, SS15}.  Compared with other techniques for satisfying DP, SVT has the unique quality that one can output some query answers without apparently paying any privacy cost.
More specifically, in SVT one is given a sequence of queries and a certain threshold $T$, and outputs a vector indicating whether each query answer is above or below $T$; that is, the output is a vector $\{\nega,\posi\}^\ell$, where $\ell$ is the number of queries answered, $\posi$ indicates that the corresponding query answer is above the threshold\update{}{,}{When there is a list of things, I prefer the syntax: a, b, and c.  This is recommended, and I also feel that the comma before and makes the list clearer, especially when b includes the word and} and $\nega$ indicates below.  SVT works by first perturbing the threshold $T$ and then comparing each perturbed individual query answer against the noisy threshold.  When one expects that the predominant majority of queries are on one side, e.g., below the threshold, one can use SVT so that while each output of $\posi$ (which we call a \textbf{positive outcome}) consumes some privacy budget, each output of $\nega$ (\textbf{negative outcome}) consumes none.  That is, with a fixed privacy budget and a given level of noise added to each query answer, one can keep answering queries as long as the number of $\posi$'s does not exceed a pre-defined cutoff point.

This ability to avoid using any privacy budget for queries with negative outcomes is very powerful for \textbf{the interactive setting}, where one answers a sequence of queries without knowing ahead of the time what these queries are.  Some well-known lower-bound results~\cite{DN03,DMNS06,DMT07,DY08} suggest that ``one cannot answer a linear, in the database size, number of queries with small noise while preserving privacy''~\cite{DNR+09}.  This limitation can be bypassed using SVT, as in the iterative construction approach in~\cite{Gupta2012, HR10, RR10}.  In this approach, one maintains a history of past queries and answers.  For each new query, one first uses this history to derive an answer for the query, and then uses SVT to check whether the error of this derived answer is below a threshold.  If it is, then one can use this derived answer for this new query without consuming any privacy budget.  Only when the error of this derived answer is above the threshold, would one need to spend privacy budget accessing the database to answer the query.

With the power of SVT come  the subtlety of why it is private and the difficulty of applying it correctly.  The version of SVT used in~\cite{Gupta2012,HR10}, which was abstracted into a generic technique and described in Roth's 2011 lecture notes~\cite{Roth2011}, turned out to be  not differentially private as claimed.  This error in~\cite{Gupta2012,HR10} is arguably not critical because it is possible to use a fixed version of SVT without affecting the main asymptotic results.
Since 2014, several variants of SVT were developed; they were used for frequent itemset mining~\cite{LC14}, for feature selection in private classification~\cite{SCM14}, and for publishing high-dimensional data~\cite{CXZX15}.  These usages are in the \textbf{non-interactive setting}, where all the queries are known ahead of the time, and the goal is to find $c$ queries that have large answers, e.g., finding the $c$ most frequent itemsets. Unfortunately, these variants do not satisfy DP, as pointed out in~\cite{CM15}.  When using a correct version of SVT in these papers, one would get significantly worse accuracy.  Since these papers seek to improve the tradeoff between privacy and utility, the results in them are thus invalid.

The fact that many usages of SVT are not private, even when proofs of their privacy were given, is already known~\cite{CM15,PrivTree}; however, we feel that what led to the erroneous proofs were not clearly explained, and such an explanation can help researchers to avoid similar errors in the future.  One evidence of the continuing confusion over SVT appears in~\cite{CM15}, the first paper that identifies errors in some SVT variants.  In~\cite{CM15}, the SVT variants in \cite{LC14, SCM14, CXZX15} were modeled as a generalized private threshold testing algorithm (GPTT), and a proof showing that GPTT does not satisfy $\epsilon$-DP for any finite $\epsilon$ (which we use $\infty$-DP to denote in this paper) was given.  However, as we show in this paper, the proof in~\cite{CM15} was incorrect.  This error was not reported in the literature. 
One goal of this paper is to clearly explain why correct usages of SVT is private, and what are the most likely confusions that caused the myriad of incorrect usages of SVT. 


A second goal of this paper is to improve the accuracy of SVT.  A version of SVT with a correct privacy proof appeared in Dwork and Roth's 2014 book~\cite{DPBook}, and was used in some recent work, e.g., \cite{SS15}.
In this paper, we present a version of SVT that adds less noise for the same level of privacy.  In addition, we develop a novel technique that optimizes the privacy budget allocation between that for perturbing the threshold and that for perturbing the query answers, and experimentally demonstrate its effectiveness.

A third goal of this paper is to point out that usage of SVT can be replaced by the Exponential Mechanism (EM)~\cite{MT07} when used in the non-interactive setting.  Most recent usages of SVT in \cite{CXZX15, LC14, SS15, SCM14} are in the non-interactive setting, where the goal is to select up to $c$ queries with the highest answers.  In this setting, one could also use the Exponential Mechanism (EM)~\cite{MT07} $c$ times to achieve the same objective, each time selecting the query with the highest answer. 
Using analysis as well as experiments, we demonstrate that EM outperforms SVT.  




In summary, this paper has the following novel contributions.
\begin{enumerate}
 \item
We propose a new version of SVT that provides better utility.  We also introduce an effective technique to improve the performance of SVT.
These enhancements achieve better utility than previous SVT algorithms and can be applied to improve utility in the interactive setting.

 \item
While previous papers have pointed out most of the errors in usages of SVT, we use a detailed privacy proof of SVT to identify the misunderstandings that likely caused the different non-private versions.  We also point out a previously unknown error in the proof in~\cite{CM15} of the non-privacy of some SVT variants.


 \item
Through analysis and experiments on real datasets, we have evaluated the effects of various SVT optimizations and compared them to EM.
Our results show that for non-interactive settings, one should use EM instead of SVT.
\end{enumerate}

The rest of the paper is organized as follows. Section~\ref{sec:background} gives background information on DP.  We analyze six variants of SVT in Section~\ref{sec:usage}.  In Section~\ref{sec:non_numeric}, we present our optimizations of SVT.  We compare SVT with the exponential mechanism in Section ~\ref{sec:svtvsem}. The experimental results are shown in Section~\ref{sec:experiments}.  Related works are summarized in Section~\ref{sec:related}.   Section~\ref{sec:conclusion} concludes our work.

\section{Background}\label{sec:background}


\begin{definition}[{$\epsilon$-DP~\cite{Dwo06,DMNS06}}] \label{def:diff}
A randomized mechanism $\AA$ satisfies $\epsilon$-differential privacy ($\epsilon$-DP) if for any pair of neighboring datasets $D$ and $D^{\prime}$, and any $S\in \mathit{Range}(\AA)$,
$$\Pr{\AA(D)=S} \leq e^{\epsilon}\cdot \Pr{\AA(D^{\prime})=S}.$$
\end{definition}

Typically, two datasets $D$ and $D^{\prime}$ are considered to be neighbors when they differ by only one tuple.
We use $D\simeq D^{\prime}$ to denote this.

There are several primitives for satisfying $\epsilon$-DP.  The Laplacian mechanism~\cite{DMNS06} adds a random noise sampled from the Laplace distribution with the scale parameter proportional to $\Delta_f$, the \emph{global sensitivity} of the function $f$.  That is, to compute $f$ on a dataset $D$, one outputs

\vspace*{-0.2in}$$
\begin{array}{crl}
& \AA_f(D) & =f(D)+\mathsf{Lap}\left(\frac{\Delta_f}{\epsilon}\right),\\
\mbox{where} &  \Delta_f & = \max\limits_{D\simeq D^{\prime}} |f(D) - f(D^{\prime})|,
\\
\mbox{and}& \Pr{\Lap{\beta}=x} & = \frac{1}{2\beta} \myexp{-|x|/\beta}.
\end{array}
$$
In the above, $\Lap{\beta}$ denotes a random variable sampled from the Laplace distribution with scale parameter $\beta$.

The \emph{exponential mechanism}~\cite{MT07} samples the output of the data analysis mechanism according to an exponential distribution.  The mechanism relies on a \emph{quality} function $q: \mathcal{D} \times \mathcal{R} \rightarrow \mathbb{R}$ that assigns a real valued score to one output $r\in\mathcal{R}$ when the input dataset is $D$, where higher scores indicate more desirable outputs.  Given the quality function $q$, its global sensitivity $\Delta_q$ is defined as:
\begin{equation*}
\Delta_q = \max_{r} \max_{D\simeq D^{\prime}} |q(D,r) - q(D^{\prime},r)|.
\end{equation*}

\noindent Outputting $r$ using the following distribution satisfies $\epsilon$-DP:
\begin{equation*}
 \Pr{r\mbox{ is selected}} \propto \exp{\left(\frac{\epsilon }{2\,\Delta_q}q(D,r)\right)}.  \label{eq:exp}
\end{equation*}

In some cases, the changes of all quality values are one-directional.  For example, this is the case when the quality function counts the number of tuples that satisfy a certain condition, and two datasets are considered to be neighboring when one is resulted from adding or deleting a tuple from the other.  When adding one tuple, all quality values either stay unchanged or increase by one; the situation where one quality increases by 1 and another decreases by 1 cannot occur.
In this case, one can make more accurate selection by choosing each possible output with probability proportional to $\exp{\left(\frac{\epsilon }{\Delta_q}q(D,r)\right)}$, instead of $\exp{\left(\frac{\epsilon }{2\,\Delta_q}q(D,r)\right)}$.

DP is sequentially composable in the sense that combining multiple mechanisms $\AA_{1}, \cdots, \AA_{m}$ that satisfy DP for $\epsilon_1, \cdots,\epsilon_m$ results in a mechanism that satisfies $\epsilon$-DP for $\epsilon=\sum_{i} \epsilon_i$.
Because of this, we refer to $\epsilon$ as the privacy budget of a privacy-preserving data analysis task.
When a task involves multiple steps, each step uses a portion of $\epsilon$ so that the sum of these portions is no more than $\epsilon$.

\section{Variants of SVT}\label{sec:usage}

\begin{figure*}[!]
\begin{centering}
    \caption{A Selection of SVT Variants} \label{fig:algorithms}
\end{centering}
    \textbf{Input/Output shared by all SVT Algorithms}

    \textbf{Input:} A private database $D$, a stream of queries $Q=q_1,q_2,\cdots$ each with sensitivity no more than $\Delta$, either a sequence of thresholds $\mathbf{T}=T_1,T_2,\cdots$ or a single threshold $T$ (see footnote $^{*}$), and $c$, the maximum number of queries to be answered with $\posi$.

    \textbf{Output:} A stream of answers $a_1, a_2, \cdots$, where each $a_i \in \{\posi,\nega\}\cup \mathbb{R}$ and $\mathbb{R}$ denotes the set of all real numbers.

\begin{minipage}[t]{8.5cm}
    \vspace{0pt}
    \begin{algorithm}[H]
    \caption{An instantiation of the SVT proposed in this paper.}
    \label{alg:sparseour}
    \begin{algorithmic}[1]
        \Require $D,Q,\Delta,\mathbf{T}=T_1,T_2,\cdots,c$.
        \State $\epsilon_1=\epsilon/2,\;\;\rho = \Lap{\Delta/\epsilon_1}$  
        \State $\epsilon_2=\epsilon - \epsilon_1,$ $\;\;$count = 0
        \For {each query $q_i \in Q$}
        \State  $\nu_i = \Lap{2c\Delta/\epsilon_2}$
        \If{$q_i(D)+\nu_i \ge T_i+\rho$}
        \State Output $a_i = \posi$
        \State count = count + 1,   {\bf Abort} if count $\ge c$.
        \Else \State Output $a_i = \nega$
        \EndIf
        \EndFor
    \end{algorithmic}
    \end{algorithm}
\end{minipage} \;\;\;%
\begin{minipage}[t]{8.5cm}
    \vspace{0pt}
    \begin{algorithm}[H]
    \caption{SVT in Dwork and Roth 2014~\protect\cite{DPBook}. }
    \label{alg:sparsebook}
    \begin{algorithmic}[1]
        \Require $D,Q,\Delta,T,c$.
        \State $\epsilon_1=\epsilon/2,\;\;$ $\rho = \Lap{c\Delta/\epsilon_1}$  
        \State $\epsilon_2=\epsilon - \epsilon_1,\;\;$ count = 0
        \For {each query $q_i \in Q$}
        \State $\nu_i = \Lap{2c\Delta/\epsilon_1}$
        \If{$q_i(D)+\nu_i \ge T+\rho$}
        \State Output $a_i = \posi$, $\rho = \Lap{c\Delta/\epsilon_2}$
        \State count = count + 1, {\bf Abort} if count $\ge c$.
        \Else \State Output $a_i = \nega$
        \EndIf
        \EndFor
    \end{algorithmic}
    \end{algorithm}
\end{minipage}

\begin{minipage}[t]{8.5cm}
    \vspace{0pt}
    \begin{algorithm}[H]
    \caption{SVT in Roth's 2011 Lecture Notes~\protect\cite{Roth2011}.  }
    \label{alg:sparselecture}
    \begin{algorithmic}[1]
        \Require $D,Q,\Delta,T,c$.
        \State $\epsilon_1=\epsilon/2,\;\;\rho = \Lap{\Delta/\epsilon_1}$,
        \State $\epsilon_2=\epsilon - \epsilon_1,\;\;$  count = 0
        \For {each query $q_i \in Q$}
        \State $\nu_i = \Lap{c\Delta/\epsilon_2}$
        \If{$q_i(D)+\nu_i \ge T + \rho$}
        \State Output $a_i = q_i(D)+\nu_i$
        \State count = count + 1,  {\bf Abort} if count $\ge c$.
        \Else \State Output $a_i = \nega$
        \EndIf
        \EndFor
    \end{algorithmic}
    \end{algorithm}
\end{minipage} %
\begin{minipage}[t]{8.5cm}
    \vspace{0pt}
    \begin{algorithm}[H]
    \caption{SVT in Lee and Clifton 2014~\protect\cite{LC14}. }
    \label{alg:LC}
    \begin{algorithmic}[1]
        \Require $D,Q,\Delta,T,c$.
        \State $\epsilon_1=\epsilon/4,\;\;\rho = \Lap{\Delta/\epsilon_1}$
        \State  $\epsilon_2=\epsilon - \epsilon_1,\;\;$ count = 0
        \For {each query $q_i \in Q$}
        \State $\nu_i = \Lap{\Delta/\epsilon_2}$
        \If{$q_i(D)+\nu_i \ge T + \rho$}
        \State Output $a_i = \posi$
        \State count = count + 1, {\bf Abort} if count $\ge c$.
        \Else       \State Output $a_i = \nega$
        \EndIf
        \EndFor
    \end{algorithmic}
    \end{algorithm}
\end{minipage}

\begin{minipage}[t]{8.5cm}
    \vspace{0pt}
    \begin{algorithm}[H]
    \caption{SVT in Stoddard et al.~2014~\protect\cite{SCM14}. }
    \label{alg:SCM}
    \begin{algorithmic}[1]
        \Require $D,Q,\Delta,T$.
        \State $\epsilon_1=\epsilon/2,\;\;\rho = \Lap{\Delta/\epsilon_1}$ 
        \State  $\epsilon_2=\epsilon - \epsilon_1$
        \For {each query $q_i \in Q$}
        \State $\nu_i=0$
        \If{$q_i(D)+\nu_i\ge T + \rho$}
        \State Output $a_i = \posi$
        \State $\;$
        \Else        \State Output $a_i = \nega$
        \EndIf
        \EndFor
    \end{algorithmic}
    \end{algorithm}
\end{minipage}
\begin{minipage}[t]{8.5cm}
    \vspace{0pt}
    \begin{algorithm}[H]
    \caption{SVT in Chen et al.~2015~\protect\cite{CXZX15}.}
    \label{alg:CXZX}
    \begin{algorithmic}[1]
        \Require $D,Q,\Delta,\mathbf{T}=T_1,T_2,\cdots$.
        \State $\epsilon_1=\epsilon/2,\;\;\rho = \Lap{\Delta/\epsilon_1}$
        \State $\epsilon_2=\epsilon - \epsilon_1$
        \For {each query $q_i \in Q$}
        \State $\nu_i = \Lap{\Delta/\epsilon_2}$
        \If{$q_i(D)+\nu_i \ge T_i + \rho$}
        \State Output $a_i = \posi$
        \State $\;$
        \Else         \State Output $a_i = \nega$
        \EndIf
        \EndFor
    \end{algorithmic}
    \end{algorithm}
\end{minipage} %

    \begin{center}
        \begin{tabular}{| c | c | c | c | c | c |  c | }
            \hline
                & \textbf{Alg.~\ref{alg:sparseour}} & \textbf{Alg.~\ref{alg:sparsebook}} & \textbf{Alg.~\ref{alg:sparselecture}} & \textbf{Alg.~\ref{alg:LC}}& \textbf{Alg.~\ref{alg:SCM}} & \textbf{Alg.~\ref{alg:CXZX}} \\
            \hline
           $\epsilon_1$& $\epsilon/2$ & $\epsilon/2$ & $\epsilon/2$ & $\epsilon/4$ & $\epsilon/2$ & $\epsilon/2$\\
            \hline
            Scale of threshold noise $\rho$ & $\Delta/\epsilon_1$ & $c\Delta/\epsilon_1$ & $\Delta/\epsilon_1$ & $\Delta/\epsilon_1$ & $\Delta/\epsilon_1$ & $\Delta/\epsilon_1$  \\
             \hline
              Reset $\rho$ after each output of $\posi$ \textbf{(unnecessary)} &  & Yes &  &  &  &  \\
             \hline
            Scale of query noise $\nu_i$  & $2c\Delta/\epsilon_2$ & $2c\Delta/\epsilon_2$ & $c\Delta/\epsilon_1$ & $\Delta/\epsilon_2$ & 0 & $\Delta/\epsilon_2$ \\
            \hline
            Outputting $q_i+\nu_i$ instead of $\posi$ \textbf{(not private)} &  &  & Yes &  &  &  \\
             \hline
            Outputting unbounded $\posi$'s \textbf{(not private)} &  &   &  &  & Yes & Yes \\
            \hline
            \textbf{Privacy Property} & $\epsilon$-DP & $\epsilon$-DP & $\infty$-DP &$\left(\frac{1+6c}{4}\epsilon\right)$-DP  & $\infty$-DP & $\infty$-DP  \\
             \hline
        \end{tabular}
    \end{center}
    \vspace*{-0.1in}\caption{Differences among Algorithms~\ref{alg:sparseour}-\ref{alg:CXZX}. }
    \label{tab:differences of algs}

\vspace*{0.15in}\begin{minipage}{172mm}

$^{*}$  Algorithms~\ref{alg:sparseour} and \ref{alg:CXZX} use a sequence of thresholds $\mathbf{T}=T_1,T_2,\cdots$, allowing different thresholds for different queries.  The other algorithms use the same threshold $T$ for all queries.  We point out that this difference is mostly syntactical.  In fact, having an SVT where the threshold always equals $0$ suffices.  Given a sequence of queries $q_1,q_2,\cdots$, and a sequence of thresholds $\mathbf{T}=T_1,T_2,\cdots$, we can define a new sequence of queries $r_i=q_i-T_i$, and apply the SVT to $r_i$ using $0$ as the threshold to obtain the same result.  In this paper, we decide to use thresholds to be consistent with the existing papers.


\end{minipage}

\end{figure*}

In this section, we analyze variants of SVT; six of them are listed in Figure~\ref{fig:algorithms}.
Alg.~\ref{alg:sparseour} is an instantiation of our proposed SVT. Alg.~\ref{alg:sparsebook} is the version taken from \cite{DPBook}.
Alg.~\ref{alg:sparselecture}, \ref{alg:LC},  \ref{alg:SCM}, and \ref{alg:CXZX} are taken from
 \cite{Roth2011,LC14,SCM14,CXZX15}
respectively.

The table in Figure \ref{tab:differences of algs} summarizes the differences among these algorithms.  Their privacy properties are given in the last row of the table.  Alg.~\ref{alg:sparseour} and~\ref{alg:sparsebook} satisfy $\epsilon$-DP, and the rest of them do not.  Alg.~\ref{alg:sparselecture}, \ref{alg:SCM}, \ref{alg:CXZX} do not satisfy $\epsilon$-DP for any finite $\epsilon$, which we denote as $\infty$-DP.

An important input parameter to any SVT algorithm is the number $c$, i.e., how many positive outcomes one can answer before stopping.  This number can be quite large.  For example, in privately finding top-$c$ frequent itemsets~\cite{LC14}, $c$ ranges from 50 to 400.  In using selective stochastic gradient descent to train deep learning model privately~\cite{SS15}, the number of gradients to upload at each epoch ranges from 15 to 140,106.  

To understand the differences between these variants, one can view  SVT as having the following four steps steps: 
\begin{ecompact}
 \item
Generate the threshold noise $\rho$ (Line 1 in each algorithm),  which will be added to the threshold during comparison between each query and the threshold (line 5).  In all except Alg.~\ref{alg:sparsebook}, $\rho$ scales with $\Delta/\epsilon_1$.  In Alg.~\ref{alg:sparsebook}, however, $\rho$ scales with $c\Delta/\epsilon_1$.  This extra factor of $c$ in the noise scale causes Alg.~\ref{alg:sparsebook} to be much less accurate than Alg.~\ref{alg:sparseour}.  We show that including the factor of $c$ is an effect of  Alg.~\ref{alg:sparsebook}'s design to resample $\rho$ each time a query results in a positive outcome (Line 6).  When keeping $\rho$ unchanged, $\rho$ does not need to scale with $c$ to achieve privacy.
 \item
For each query $q_i$, generate noise $\nu_i$ to be added to the query (Line 4), which should scale with $2c\Delta/\epsilon_2$.  In Alg.~\ref{alg:LC} and \ref{alg:CXZX}, $\nu_i$ scales with $\Delta/\epsilon_2$.  Removing the factor of $c$ from the magnitude of the noise will result in better utility; however, this is done at the cost of being non-private.   Alg.~\ref{alg:SCM} adds no noise to $q_i$ at all, and is also non-private.
 \item
Compare the perturbed query answer with the noisy threshold and output whether it is above or below the threshold (Lines 5, 6, 9).  Here Alg.~\ref{alg:sparseour} differs in that it outputs the noisy query answer $q_i(D)+\nu_i$, instead of an indicator $\posi$.  This makes it non-private.
 \item
Keep track of the number of $\posi$'s in the output, and stop when one has outputted $c$ $\posi$'s (Line 7).  This step is missed in Alg.~\ref{alg:SCM} and \ref{alg:CXZX}.  Without this limitation, one can answer as many queries as there are with a fixed accuracy level for each query.  If this was to be private, then one obtains privacy kind of ``for free''.
\end{ecompact}

\subsection{Privacy Proof for Alg.~\ref{alg:sparseour}}

We now prove the privacy of Alg.~\ref{alg:sparseour}. We break down the proof into two steps, to make the proof easier to understand, and, more importantly, to enable us to point out what confusions likely cause the different non-private variants of SVT to be proposed.  In the first step, we analyze the situation where the output is $\nega^\ell$, a length-$\ell$ vector $\langle\nega,\cdots,\nega\rangle$, indicating that all $\ell$ queries are tested to be below the threshold.

\begin{lemma}\label{lemma:negative}
Let $\AA$ be Alg.~\ref{alg:sparseour}.  For any neighboring datasets $D$ and $D^{\prime}$, and any integer $\ell$, we have
\begin{equation*}
 \Pr{\AA(D)=\nega^\ell} \leq e^{\epsilon_1} \Pr{\AA(D^{\prime})=\nega^\ell}. \label{eq:nega}
\end{equation*}
\end{lemma}

\begin{proof}
We have
{\small
\begin{align}
\nonumber \Pr{\AA(D) =\nega^\ell} &=  \int_{-\infty}^{\infty} \Pr{\rho=z} f_{D}(z)\: dz,
\\
\mbox{where }\;f_{D}(z) & = \Pr{\AA(D) =\nega^\ell \mid \rho=z}  \label{eq:con1}
\\
 & =  \prod\limits_{i \in \{1, 2, \cdots, \ell\}}\Pr{ q_i(D) + \nu_i < T_i + z}. \label{eq:con2}
\end{align}
}

\vspace*{-0.1in}\noindent
The probability of outputting $\nega^\ell$ over $D$ is the summation (or integral) of the product of $\Pr{\rho=z}$, the probability that the threshold noise equals $z$, and $f_{D}(z)$, the conditional probability
that $\nega^\ell$ is the output on $D$ given that the threshold noise $\rho$ is $z$.
The step from (\ref{eq:con1}) to (\ref{eq:con2}) is because, given $D$, $\mathbf{T}$, the queries, and $\rho$, whether one query results in $\nega$ or not depends completely on the noise $\nu_i$ and is independent from whether any other query results in $\nega$.

The key observation underlying the SVT technique is that for any neighboring $D,D'$, we have $f_{D}(z) \leq f_{D'}(z+\Delta)$.  Suppose that we have $q_i(D)=q_i(D')-\Delta$ for each $q_i$, then the ratio $f_{D}(z)/f_{D'}(z)$ is unbounded when $|L|$ is unbounded.  However, $f_{D}(z)$ is upper-bounded by the case where the dataset is $D'$ but the noisy threshold  \textbf{is increased by} $\Delta$, because for any query $q_i$, $|q_i(D)-q_i(D')|\leq \Delta$.  More precisely, we have
{\small
\begin{align}
\nonumber
 \Pr{ q_i(D) + \nu_i < T_i+z} & = \Pr{ \nu_i < T_i - q_i(D) + z} \\
 \nonumber
 & \leq \Pr{ \nu_i < T_i + \Delta - q_i(D^{\prime}) + z} \\
 & = \Pr{ q_i(D^{\prime}) + \nu_i < T_i + (z + \Delta)}.  \label{eq:prneg}
\end{align}
}
Because $\rho=\Lap{\Delta/\epsilon_1}$, by the property of the Laplace distribution,  we have:
\begin{align}\small
\forall z,\;\Pr{\rho=z} & \leq e^{\epsilon_1}\,\Pr{\rho=z+\Delta}, \mbox{ and thus}  \nonumber
\end{align}
{\small
\begin{align*}
\Pr{\AA(D) =\nega^\ell} &=  \int_{-\infty}^{\infty}  \Pr{\rho=z} \; f_{D}(z)\: dz \\
 & \leq \int_{-\infty}^{\infty} e^{\epsilon_1} \Pr{\rho=z+\Delta} \; f_{D'}(z+\Delta)\: dz  \\ 
 & =  e^{\epsilon_1} \int_{-\infty}^{\infty} \Pr{\rho=z'} f_{D'} (z^{\prime})\: dz^{\prime}    \qquad  {\rm let}\: z^{\prime}=z+\Delta\\
 & =   e^{\epsilon_1} \Pr{\AA(D^{\prime}) =\nega^\ell}.
\end{align*}
}
This proves the lemma.
\end{proof}

We can obtain a similar result when the output is $\posi^\ell$ instead of $\nega^\ell$, i.e., $\Pr{\AA(D)=\posi^\ell}\,  \leq\, e^{\epsilon_1} \Pr{\AA(D^{\prime})=\posi^\ell}$,
because $\Pr{\rho=z} \leq e^{\epsilon_1}\,\Pr{\rho=z-\Delta}$ and  $g_{D}(z) \leq g_{D'}(z-\Delta)$, where
\begin{align}
 g_{D}(z)\: & =\:  \prod_{ i}\Pr{ q_i(D) + \nu_i \ge T_i + z}.  \label{eq:posi}
\end{align}

The fact that this bounding technique works both for positive outputs and negative outputs likely contributes to the misunderstandings behind Alg.~\ref{alg:SCM} and~\ref{alg:CXZX}, which treat positive and negative outputs exactly the same way.  The error is that when the output consists of both $\nega$ and $\posi$, one has to choose one side (either positive or negative) to be bounded by the above technique, and cannot do both at the same time.

We also observe that the proof of Lemma~\ref{lemma:negative} will go through if no noise is added to the query answers, i.e., $\nu_i=0$, because   Eq~(\ref{eq:prneg}) holds even when $\nu_i=0$.  It is likely because of this observation that Alg.~\ref{alg:SCM} adds no noise to query answers.
However, when considering  outcomes that include both positive answers ($\posi$'s) and negative answers ($\nega$'s), one has to add noises to the query answers, as we show below.


\begin{theorem}
\label{thm:algsparseour}
Alg.~\ref{alg:sparseour} is $\epsilon$-DP.
\end{theorem}
\begin{proof}
Consider any output vector $\vec{a} \in \{\nega, \posi\}^\ell$.  Let $\vec{a}=\langle a_1,\cdots,a_\ell\rangle$, $\IPOSI=\{i: a_i =\posi\}$, and $\INEGA=\{i: a_i = \nega \}$.    Clearly, $|\IPOSI|\le c$.  We have
{\small
\begin{align}
\Pr{\AA(D) =\vec{a}} & =  \int_{-\infty}^{\infty} \Pr{\rho\!=\!z}\;\; f_{D}(z)\;\; g_{D}(z) \: dz, \label{eq:prgen}
 \\
 \;\mbox{where}\;\;\; f_{D}(z) & = \prod\limits_{i \in \INEGA } \Pr{q_i(D)\splus\nu_i\!<\! T_i \splus z} \nonumber
 \\
 \;\mbox{and}\;\;\; g_{D}(z) & = \prod\limits_{i \in \IPOSI } \Pr{q_i(D)\splus \nu_i\!\ge\! T_i \splus z}. \nonumber
\end{align}
}
\noindent The following, together with $\epsilon=\epsilon_1+\epsilon_2$, prove this theorem:
\begin{align}
\Pr{\rho\!=\!z} & \leq e^{\epsilon_1} \Pr{\rho\!=\!z + \Delta}  \nonumber
 \\
  f_{D}(z) & \leq f_{D'}(z + \Delta)  \label{eq:fprop}
 \\
  g_{D}(z) & \leq e^{\epsilon_2} g_{D'}(z + \Delta).  \label{eq:gprop}
\end{align}
Eq.~(\ref{eq:fprop}) deals with all the negative outcomes.  Eq.~(\ref{eq:gprop}), which deals with positive outcomes, is ensured by several factors.  At most $c$ positive outcomes can occur, $|q_i(D)-q_i(D^{\prime})|\leq \Delta$, and the threshold for $D'$ is just $\Delta$ higher that for $D$; thus adding noise $\nu_i=\Lap{2c\Delta/\epsilon_2}$ to each query ensures the desired bound.  More precisely,
\vspace*{-0.06in}{\small
\begin{align}
g_{D}(z)
  & = \prod\limits_{i \in \IPOSI } \Pr{ \nu_i  \ge  T_i \splus z \sminus q_i(D) } \nonumber \\
 & \leq \prod\limits_{i \in \IPOSI } \Pr{ \nu_i \ge  T_i \splus z - \Delta - q_i(D^{\prime})} \label{eq:rdelta}\\
 & \leq \prod\limits_{i \in \IPOSI } e^{\epsilon_2/c} \Pr{ \nu_i \ge  T_i\splus z \sminus \Delta \sminus q_i(D^{\prime}) \splus 2\Delta} \label{eq:lapshift} \\
 & \le e^{\epsilon_2} \prod\limits_{i \in \IPOSI } \Pr{q_i(D^{\prime})+\nu_i\ge T_i+z+\Delta} \label{eq:atmostc}
 \\
 &  = e^{\epsilon_2} g_{D'}(z+\Delta)\nonumber.
\end{align}
}
Eq.~(\ref{eq:rdelta}) is because {\small $ \sminus q_i(D) \geq  \sminus \Delta \sminus q_i(D^{\prime})$}, Eq.~(\ref{eq:lapshift}) is from the Laplace distribution's property, and
Eq.~(\ref{eq:atmostc}) is because there are at most $c$ positive outcomes, i.e., $|\IPOSI|\leq c$.
\end{proof}

We observe that while  $g_{D}(z) \leq g_{D'}(z - \Delta)$ is true, replacing (\ref{eq:gprop}) with it does not help us prove anything, because (\ref{eq:fprop}) uses $(z + \Delta)$ and (\ref{eq:gprop}) uses $(z - \Delta)$, and we cannot change the integration variable in a consistent way.

\subsection{Privacy Properties of Other Variants}

\vspace*{0.05in}\noindent\textbf{Alg.~\ref{alg:sparsebook}} is taken from the differential privacy book published in 2014~\cite{DPBook}.  It satisfies $\epsilon$-DP.  It has two differences when compared with Alg.~\ref{alg:sparseour}.  First, $\rho$ follows $\Lap{c\Delta/\epsilon_1}$ instead of $\Lap{\Delta/\epsilon_1}$.  This causes Alg.~\ref{alg:sparsebook} to have significantly worse performance than Alg.~\ref{alg:sparseour}, as we show in Section~\ref{sec:experiments}.
Second, Alg.~\ref{alg:sparsebook} refreshes the noisy threshold $T$ after each output of $\posi$.  We note that making the threshold noise scale with $c$ is necessary for privacy \emph{only if} one refreshes the threshold noise after each output of $\posi$; however, such refreshing is unnecessary.

\vspace*{0.05in}\noindent\textbf{Alg.~\ref{alg:sparselecture}} is taken from~\cite{Roth2011}, which in turn was abstracted from the algorithms used in~\cite{Gupta2012,HR10}.  It has two differences from Alg.~\ref{alg:sparseour}.  First, $\nu_i$ follows $\Lap{c\Delta/\epsilon_2}$ instead of $\Lap{2c\Delta/\epsilon_1}$; this is not enough for $\epsilon$-DP (even though it suffices for $\frac{3\epsilon}{2}$-DP).   Second, it actually outputs the noisy query answer instead of  $\posi$ for a query  above the threshold.
%
This latter fact causes Alg.~\ref{alg:sparselecture} to be not $\epsilon^{\prime}$-DP for any finite $\epsilon^{\prime}$.  A proof for this appeared in Appendix~A of~\cite{PrivTree};
\refertoappendix{we omit it here because of space limitation. }{for completeness, see Appendix~\cite{app:proofs} for the proof.}
The error in the proof for Alg.~\ref{alg:sparselecture}'s privacy in \cite{Roth2011} occurs in the following steps:
{\scriptsize
\begin{align}
 \nonumber
 & \Pr{\AA(D)=\vec{a}}
\\
 =& \int_{-\infty}^{\infty} \!\Pr{\rho\!=\!z} f_{D}(z) \prod\limits_{i \in \IPOSI} \Pr{q_i(D)\splus \nu_i\!\ge\! T\splus z  \wedge q_i(D)\splus \nu_i \!=\!a_i}  \: dz \nonumber \\
 = &  \int_{-\infty}^{\infty} \! \Pr{\rho\!=\!z} f_{D}(z) \prod\limits_{i \in \IPOSI} \Pr{q_i(D)\splus \nu_i =a_i}  \: dz  \label{eq:Rothproof2}
 \\
 \le & \int_{-\infty}^\infty  e^{\epsilon_1}  \Pr{\rho\!=\!z+\Delta} f_{D'}(z+\Delta)\: dz \prod\limits_{i \in \IPOSI} e^{\epsilon_2/c}\Pr{q_i(D^{\prime})+\nu_i =a_i}   \nonumber 
\end{align}
}
The error occurs when going to (\ref{eq:Rothproof2}), which is implicitly done in~\cite{Roth2011}.
This step removes the condition $q_i(D)\splus\nu_i\ge T\splus z$.

Another way to look at this error is that outputting the positive query answers reveals information about the noisy threshold, since the noisy threshold must be below the outputted query answer.  Once information about the noisy threshold is leaked, the ability to answer each negative query ``for free'' disappears.

\vspace*{0.05in}\noindent\textbf{Alg.~\ref{alg:LC}}, taken from~\cite{LC14}, differs from Alg.~\ref{alg:sparseour} in the following ways.  First, it sets $\epsilon_1$ to be $\epsilon/4$ instead of $\epsilon/2$.  This has no impact on the privacy.  Second, $\nu_i$ does not scale with $c$.
As a result, Alg.~\ref{alg:LC}  is only $\left(\frac{1+6c}{4}\right)\epsilon$-DP in general.
In~\cite{LC14}, Alg.~\ref{alg:LC} is applied for finding frequent itemsets, where the queries are counting queries and are monotonic.  Because of this monotonicity, the usage of  Alg.~\ref{alg:LC}  here is $\left(\frac{1+3c}{4}\right)\epsilon$-DP. Theorem \ref{standardthm} can be applied to Alg.~\ref{alg:LC} to establish this privacy property; we thus omit the proof of this.

\vspace*{0.05in}\noindent\textbf{Alg.~\ref{alg:CXZX}}, taken from~\cite{CXZX15}, was motivated by the observation that the proof in~\cite{LC14} can go through without stopping after encountering $c$ positive outcomes, and removed this limitation.

\vspace*{0.05in}\noindent\textbf{Alg~\ref{alg:SCM}}, taken from~\cite{SCM14}, further used the observation that the derivation of Lemma~\ref{lemma:negative} does not depend on the addition of noises, and removed that part as well.
The proofs for Alg.~\ref{alg:LC}, \ref{alg:SCM}, \ref{alg:CXZX} in~\cite{LC14, SCM14,CXZX15}
roughly use the logic below.
{\scriptsize
\begin{align*}
\int_{-\infty}^{\infty}\! \Pr{\rho\!=\!z} f_{D}(z) g_{D}(z)  dz
  & \leq e^{\epsilon} \int_{-\infty}^{\infty}\! \Pr{\rho\!=\!z} f_{D'}(z) g_{D'}(z)  dz
\\
\mbox{because } \int_{-\infty}^{\infty} \Pr{\rho\!=\!z} f_{D}(z)\: dz & \leq e^{\epsilon/2} \int_{-\infty}^{\infty} \Pr{\rho\!=\!z}\ f_{D'}(z)\: dz
\\
\mbox{and } \int_{-\infty}^{\infty} \Pr{\rho\!=\!z} g_{D}(z) \: dz  &\leq e^{\epsilon/2} \int_{-\infty}^{\infty} \Pr{\rho\!=\!z} g_{D'}(z),
\end{align*}
}

This logic incorrectly assumes the following is true:
{\small
$$\int_{-\infty}^{\infty}\! p(z) f(z) g(z)  dz = \int_{-\infty}^{\infty}\! p(z) f(z)  dz\: \int_{-\infty}^{\infty}\! p(z) g(z)  dz$$
}

A proof that Alg.~\ref{alg:CXZX} does not satisfy $\epsilon$-DP for any finite $\epsilon$ is given in Appendix~B of \cite{PrivTree}.
While these proofs also apply to Alg.~\ref{alg:SCM}, we give a much simpler proof of this below.

\begin{theorem}
\label{prop:Alg5infinite}
Alg.~\ref{alg:SCM} is not $\epsilon^{\prime}$-DP for any finite $\epsilon^{\prime}$.
\end{theorem}
\begin{proof}
Consider a simple example, with $T=0$, $\Delta=1$, $\mathbf{q}=\langle q_1,q_2\rangle$ such that $\mathbf{q}(D)=\langle 0, 1\rangle$ and $\mathbf{q}(D^{\prime})=\langle 1, 0 \rangle$, and $a=\langle\bot,\top\rangle$. Then by Eq (\ref{eq:prgen}),  we have
{\small
\begin{align*}
& \Pr{\AA(D)=a}= \int_{-\infty}^{\infty} \Pr{\rho=z}   \Pr{0< z} \Pr{1 \ge  z}\: dz \\
&= \int_{0}^{1} \Pr{\rho=z} \: dz > 0,
\end{align*}
}
which is nonzero; and
{\small
\begin{align*}
\Pr{\AA(D^{\prime})=a}= \int_{-\infty}^{\infty} \Pr{\rho=z^{\prime}}   \Pr{1< z^{\prime}} \Pr{0 \ge  z^{\prime}}\: dz^{\prime},
\end{align*}
}
which is zero. So the probability ratio $\frac{\Pr{\AA(D)=a}}{\Pr{\AA(D^{\prime})=a}}=\infty$.
\end{proof}

\subsection{Error in Privacy Analysis of GPTT}\label{sec:gptt}

In~\cite{CM15}, the SVT variants in \cite{LC14, SCM14, CXZX15} were modeled as a generalized private threshold testing algorithm (GPTT).  In GPTT, the threshold $T$ is perturbed using $\rho=\Lap{\Delta/\epsilon_1}$ and each query answer is perturbed using $\Lap{\Delta/\epsilon_2}$ and there is no cutoff; thus GPTT can be viewed as a generalization of Algorithm \ref{alg:CXZX}.  When setting $\epsilon_1=\epsilon_2=\frac{\epsilon}{2}$, GPTT becomes   Alg.~\ref{alg:CXZX}.

There is a constructive proof in~\cite{CM15} to show that GPTT is not $\epsilon'$-DP for any finite $\epsilon'$.  However, this proof is incorrect.  This error is quite subtle.  We discovered the error only after observing that the technique of the proof can be applied to show that Alg.~\ref{alg:sparseour} (which we have proved to be private) to be non-private.  The detailed discussion of this error is quite technical, and is included in Appendix~\ref{app:gptt}.

\subsection{Other Variants} \label{sec:other}

Some usages of SVT aim at satisfying $(\epsilon,\delta)$-DP~\cite{DMNS06}, instead of $\epsilon$-DP. These often exploit the advanced composition theorem for DP~\cite{DRV10}, which states that applying $k$ instances of $\epsilon$-DP algorithms satisfies $(\epsilon^{\prime}, \delta^{\prime})$-DP, where $\epsilon^{\prime} = \sqrt{2k\ln(1/\delta^{\prime})}\epsilon + k\epsilon (e^{\epsilon} - 1)$.  In this paper, we limit our attention to SVT variants to those satisfying $\epsilon$-DP, which are what have been used in the data mining community~\cite{CXZX15, LC14, SS15, SCM14}.

The SVT used in~\cite{HR10,RR10} has another difference from Alg.~\ref{alg:sparselecture}.  In~\cite{HR10,RR10}, the goal of using SVT is to determine whether the error of using an answer derived from past queries/answers is below a threshold.
This check takes the form of ``$\mathbf{if }\; |\tilde{q_i}- q_i(D)+\nu_i| \ge T + \rho \;\mathbf{ then}\; \mbox{output } i$,'' where $\tilde{q_i}$ gives the estimated answer of a query obtained using past queries/answers, and $q_i(D)$ gives the true answer.  This is incorrect because the noise $\nu_i$ should be outside the absolute value sign.  In the usage in~\cite{HR10,RR10}, the left hand of the comparison is always $\geq 0$; thus
whenever the output includes at least one $\posi$, one immediately knows that the threshold noise $\rho \!\geq\! -T$.  This leakage of $\rho$ is somewhat similar to  Alg.~\ref{alg:sparselecture}'s leakage caused by outputting noisy query answers that are found to be above the noisy threshold.  This problem can be fixed by using ``$\mathbf{if }\; |\tilde{q_i}-q_i(D)|+\nu_i \ge T + \rho \;\mathbf{ then}\; \mbox{output } i$'' instead.  By viewing $r_i=|\tilde{q_i}-q_i(D)|$ as the query to be answered; this becomes a standard application of SVT.

\section{Optimizing SVT}\label{sec:non_numeric}


Alg.~\ref{alg:sparseour} can be viewed as allocating half of the privacy budget for perturbing the threshold and half for perturbing the query answers.  This allocation is somewhat arbitrary, and other allocations are possible.  Indeed, Alg.~\ref{alg:LC} uses a ratio of $1:3$ instead of $1:1$.  In this section, we study how to improve SVT by optimizing this allocation ratio and by introducing other techniques.

\subsection{A Generalized SVT Algorithm}

We present a generalized SVT algorithm in Alg.~\ref{alg:standard}, which uses $\epsilon_1$ to perturb the threshold and $\epsilon_2$ to perturb the query answers. Furthermore, to accommodate the situations where one wants the noisy counts for positive queries, we also use $\epsilon_3$ to output query answers using the Laplace mechanism.

\begin{algorithm}
    \caption{Our Proposed Standard SVT}
    \label{alg:standard}
    \begin{algorithmic}[1]
        \Require $D,Q,\Delta,\mathbf{T}=T_1,T_2,\cdots,c$ and $\epsilon_1, \epsilon_2$ and $\epsilon_3$.
        \Ensure A stream of answers $a_1, a_2, \cdots$
        \State $\rho = \Lap{\frac{\Delta}{\epsilon_1}}$, count = 0
        \For {Each query $q_i \in Q$}
        \State $\nu_i = \Lap{\frac{2c\Delta}{ \epsilon_2}}$
        \If{$q_i(D)+\nu_i \ge T_{i} + \rho$}
            \If{$\epsilon_3 >0$}
                \State \textbf{Output} $a_i = q_i(D)+\Lap{\frac{c\Delta}{\epsilon_3}}$
            \Else
                \State \textbf{Output} $a_i = \top$
            \EndIf
            \State count = count + 1, {\bf Abort} if count $\ge c$.
        \Else
            \State \textbf{Output} $a_i = \bot$
        \EndIf
        \EndFor
    \end{algorithmic}
\end{algorithm}

We now prove the privacy for Alg.~\ref{alg:standard}; the proof requires only minor changes from the proof of Theorem \ref{thm:algsparseour}.

\begin{theorem}
\label{standardthm}
Alg.~\ref{alg:standard} is $(\epsilon_1+\epsilon_2+\epsilon_3)$-DP.
\end{theorem}

\begin{proof}  Alg.~\ref{alg:standard} can be  divided into two phases, the first phase  outputs a vector to mark which query is above the threshold and the second phase uses the Laplace mechanism to output noisy counts for the queries that are found to be above the threshold in the first phase.  Since the second phase is $\epsilon_3$-DP,   it suffices to show that the first phase is   $(\epsilon_1+\epsilon_2)$-DP.
 %
For any output vector $\vec{a} \in \{\posi, \nega\}^\ell$,  we want to show

{\small
\begin{align*}
\Pr{\AA(D) =\vec{a}} & =  \int_{-\infty}^{\infty} \Pr{\rho\!=\!z}\;\; f_{D}(z)\;\; g_{D}(z) \: dz        \\
                     & \le \int_{-\infty}^{\infty} e^{\epsilon_1+\epsilon_2}\Pr{\rho\!=\!z+\Delta}\; f_{D'}(z+\Delta)\; g_{D'}(z+\Delta) \: dz
\\
 &=  e^{\epsilon_1+\epsilon_2}\Pr{\AA(D^{\prime}) =\vec{a}}.
 \end{align*}
}

This holds because,  similarly to the proof of Theorem \ref{thm:algsparseour},
{\small
\begin{align*}
  \Pr{\rho\!=\!z} & \le e^{\epsilon_1}\Pr{\rho\!=\!z+\Delta}, \\
 f_{D}(z) & = \prod\limits_{i \in \INEGA } \Pr{q_i(D)\splus\nu_i\!<\! T_i \splus z} \le f_{D'}(z+\Delta),\\
 g_{D}(z) & = \prod\limits_{i \in \IPOSI } \Pr{q_i(D)\splus \nu_i\!\ge\!  T_i \splus z} \le e^{\epsilon_2} g_{D'}(z+\Delta).
\end{align*}
}
\end{proof}

\subsection{Optimizing Privacy Budget Allocation}\label{sec:budget_alloc_opt}

In Alg.~\ref{alg:standard}, one needs to decide how to divide up a total privacy budget $\epsilon$ into $\epsilon_1,\epsilon_2,\epsilon_3$.  We note that $\epsilon_1+\epsilon_2$ is used for outputting the indicator vector, and $\epsilon_3$ is used for outputting the noisy counts for queries found to be above the threshold; thus the ratio of $(\epsilon_1+\epsilon_2):\epsilon_3$ is determined by the domain needs and should be an input to the algorithm.

On the other hand, the ratio of $\epsilon_1:\epsilon_2$ affects the accuracy of SVT.  Most variants use $1:1$, without a clear justification.  To choose a ratio that can be justified, we observe that this ratio affects the accuracy of the following comparison:
{\small
$$q_i(D)+\Lap{\frac{2c\Delta}{ \epsilon_2}}  \geq T + \Lap{\frac{\Delta}{ \epsilon_1}}.$$
}
To make this comparison as accurate as possible, we want to minimize the variance of $\Lap{\frac{\Delta}{ \epsilon_1}} - \Lap{\frac{2c\Delta}{ \epsilon_2}}$, which is
\begin{equation*}
2 \left(\frac{\Delta}{\epsilon_1}\right)^2+2\left( \frac{2c\Delta}{\epsilon_2}\right)^2,
\end{equation*}
 when $\epsilon_1+\epsilon_2$ is fixed. This is minimized when
\begin{equation}
 \epsilon_1:\epsilon_2= 1:(2c)^{2/3}. \label {optimationbudget}
\end{equation}
We will evaluate the improvement resulted from this optimization in Section~\ref{sec:experiments}.

\subsection{SVT for Monotonic Queries}\label{sec:mono_queries}

In some usages of SVT, the queries are monotonic.  That is, when changing from $D$ to $D^{\prime}$, all queries whose answers are different change in the same direction, i.e., there do not exist $q_i,q_j$ such that $(q_i(D)>q_i(D^{\prime})) \wedge (q_j(D)<q_j(D^{\prime}))$.  That is, we have either $\forall_i\, q_i(D) \geq q_i(D^{\prime})$, or $\forall_i\, q_i(D') \geq q_i(D)$.
This is the case when using SVT for frequent itemset mining in~\cite{LC14} with neighboring datasets defined as adding or removing one tuple.
For monotonic queries, adding $\Lap{\frac{c\Delta}{ \epsilon_2}}$ instead of $\Lap{\frac{2c\Delta}{ \epsilon_2}}$ suffices for privacy.

\begin{theorem}
\label{non2thm}
Alg.~\ref{alg:standard}  with $\nu_i= \Lap{\frac{c\Delta}{\epsilon_2}}$ in line 3 satisfies  $(\epsilon_1+\epsilon_2+\epsilon_3) $-DP when all queries are monotonic.
\end{theorem}
\begin{proof}
Because the second phase of Alg.~\ref{alg:standard} is still $\epsilon_3$-DP, we just need to show that for any output vector $\vec{a}$,

\vspace*{-0.15in}{\small \begin{align*}
\Pr{\AA(D)=\vec{a}} & = \int_{-\infty}^{\infty} \Pr{\rho\!=\!z}\;\; f_{D}(z)\;\; g_{D}(z) \: dz  \\
                    & \le e^{\epsilon_1+\epsilon_2} \Pr{\AA(D^{\prime})=  \vec{a}},
\\
\mbox{where } \;f_{D}(z) & = \prod\limits_{i \in \INEGA } \Pr{q_i(D)\splus\nu_i\!<\! T_i \splus z},
\\
\mbox{ and } \; g_{D}(z)& = \prod\limits_{i \in \IPOSI } \Pr{q_i(D)\splus \nu_i\!\ge\!  T_i \splus z}.
\end{align*}
}
It suffices to show that either $\Pr{\rho\!=\!z}f_{D}(z)g_{D}(z) \leq e^{\epsilon_1+\epsilon_2} \Pr{\rho\!=\!z}f_{D'}(z)g_{D'}(z)$, or $\Pr{\rho\!=\!z} f_{D}(z)g_{D}(z) \leq e^{\epsilon_1+\epsilon_2} \Pr{\rho\!=\!z+\Delta}f_{D'}(z+\Delta)g_{D'}(z+\Delta)$.

First consider the case that $q_i(D) \ge q_i(D^{\prime})$ for any query $q_i$.  In this case, we have
{\small \begin{align*}
 \Pr{q_i(D) + \nu_i < T_i + z} \leq \Pr{q_i(D') + \nu_i < T_i + z},
\end{align*}
}
and thus $f_{D}(z)\le f_{D'}(z)$. Note that  $q_i(D)-q_i(D') \le \Delta$. Therefore, $g_{D}(z) \leq e^{\epsilon_2} g_{D'}(z)$, without increasing the noisy threshold by $\Delta$, because $\Pr{q_i(D)+\nu_i\ge T_i+z} \leq \Pr{q_i(D') + \nu_i \ge  T_i + z-\Delta}  \leq e^{\frac{\epsilon_2}{c}}\Pr{q_i(D')+\nu_i\ge T_i+z}$ since $\nu_i=\Lap{\frac{c\Delta}{\epsilon_2}}$.

Then consider the case in which $q_i(D) \le q_i(D^{\prime})$ for any query $q_i$.  We have the usual
{\small \begin{align*}
& f_{D}(z) \le f_{D'}(z+\Delta), \\
\mbox{ and } \;  &\Pr{\rho\!=\!z}   \le e^{\epsilon_1}\Pr{\rho\!=\!z+\Delta},
\end{align*}
}as in previous proofs. With the constraint $q_i(D) \le q_i(D^{\prime})$, using $\nu_i= \Lap{\frac{c\Delta}{\epsilon_2}}$ suffices to ensure that {\small $\Pr{q_i(D)+\nu_i\ge T_i+z} \leq e^{\frac{\epsilon_2}{c}}\Pr{q_i(D')+\nu_i\ge T_i+\Delta+z}$.}  Thus $g_{D}(z) \leq e^{\epsilon_2} g_{D'}(z+\Delta)$ holds.
\end{proof}

For monotonic queries, the  optimization  of privacy budget allocation (\ref{optimationbudget})  becomes $ \epsilon_1:\epsilon_2= 1:c^{2/3}$.
\section{SVT versus EM}\label{sec:svtvsem}


We now discuss the application of SVT in the non-interactive setting, where all the queries are known ahead of the time.  We note that most recent usages of SVT, e.g.,~\cite{CXZX15, LC14, SS15, SCM14, ZCP+14}, are in the non-interactive setting.  Furthermore, these applications of SVT aim at selecting up to $c$ queries with the highest answers.  
In~\cite{LC14}, SVT is applied to find the $c$ most frequent itemsets, where the queries are the supports for the itemsets. 
In~\cite{CXZX15}, the goal of using SVT is to determine the structure of a Bayesian Network that preserves as much information of the dataset as possible.  To this end, they select attribute groups that are highly correlated and create edges for such groups in the network.  While the algorithm in~\cite{CXZX15} takes the form of selecting attribute groups with score above a certain threshold, the real goal is to select the groups with the highest scores.
In~\cite{SS15}, SVT is used to select parameters to be shared when trying to learn neural-network models in a private fashion.  Once selected, noises are added to these parameters before they are shared.  The selection step aims at selecting the parameters with the highest scores.  


\mypara{EM or SVT.}
In non-interactive setting, one can also use the Exponential Mechanism (EM)~\cite{MT07} to achieve the same objective of selecting the top $c$ queries.  More specifically, one runs EM $c$ times, each round with privacy budget $\frac{\epsilon}{c}$.  The quality for each query is its answer; thus each query is selected with probability proportion to $\exp{\left(\frac{\epsilon }{2c\Delta}\right)}$ in the general case and to $\exp{\left(\frac{\epsilon }{c\Delta}\right)}$ in the monotonic case.  After one query is selected, it is removed from the pool of candidate queries for the remaining rounds.

An intriguing question is which of SVT and EM offers higher accuracy.
Theorem~3.24 in~\cite{DPBook} regarding the utility of SVT with $c=\Delta=1$ states: For any sequence of $k$ queries $f_1, \ldots , f_k$ such that
$|\{i < k : f_i(D) \geq T - \alpha\}| = 0$ (i.e. the only query close to being above threshold is possibly the last one), SVT
is $(\alpha, \beta)$ accurate (meaning that with probability at least $1-\beta$, all queries with answer below $T-\alpha$ result in $\nega$ and all queries with answers above $T-\alpha$ result in $\posi$) for: $\alpha_{\mathrm{SVT}} = 8(\log k +\log(2/\beta))/\epsilon $.

In the case where the last query is at least $T+\alpha$, being $(\alpha,\beta)$-correct ensures that with probability at least $1-\beta$, the correct selection is made.  For the same setting, we say that {\it EM is $(\alpha,\beta)$-correct} if given $k-1$ queries with answer $\leq T-\alpha$ and one query with answer $\ge T+\alpha$, the correct selection is made with probability at least $1-\beta$.  The probability of selecting the query with answer $\ge T+\alpha$ is at least $\frac{e^{\epsilon(T+\alpha)/2}}{(k-1)e^{\epsilon(T-\alpha)/2}+e^{\epsilon(T+\alpha)/2}}$ by the definition of EM. To ensure this probability is at least $1-\beta$,
\begin{equation*}
\alpha_{\mathrm{EM}}=(\log(k-1)+\log((1-\beta)/\beta))/\epsilon,
\end{equation*}
which is less than $1/8$ of the $\alpha_{\mathrm{SVT}}$, which suggests that EM is more accurate than SVT.

The above analysis relies on assuming that the first $k-1$ queries are no more than $T-\alpha$.  When that is not assumed, it is difficult analyze the utility of either SVT or EM.  Therefore, we will use experimental methods to compare SVT with EM.

\mypara{SVT with Retraversal.}
We want to find the most optimized version of SVT to compare with EM, and note that another interesting parameter that one can tune when applying SVT is that of the threshold $T$. When $T$ is high, the algorithm may select fewer than $c$ queries after traversing all queries.  Since roughly each selected query consumes $\frac{1}{c}$'th of the privacy budget, outputting few than $c$ queries kind of ``wasted'' the remaining privacy budget.  When $T$ is low, however, the algorithm may have selected $c$ queries before encountering later queries.  No matter how large some of these later query answers are, they cannot be selected.

We observe that in the non-interactive setting, there is a way to deal with this challenge.  One can use a higher threshold $T$, and when the algorithm runs out of queries before finding $c$ above-threshold queries, one can retraverse the list of queries that have not been selected so far, until $c$ queries are selected.  However, it is unclear how to select the optimal threshold.  In our experiments, we consider SVT-ReTr, which increases the threshold $T$ by different multiples of the scale factor of the Laplace noise injected to each query, and applies the retraversal technique.


\section{Evaluation} \label{sec:experiments}

In this section, we experimentally compare the different versions of the SVT algorithm, including our proposed SVT algorithm with different privacy budget allocation methods.  We also compare the SVT variants applicable in the non-interactive setting with EM.

\mypara{Utility Measures.}  
%
%
Since the goal of applying SVT or EM is to select the top queries, one standard metric is \emph{False negative rate (FNR)}, i.e., the fraction of true top-$c$ queries that are missed.
When an algorithm outputs exactly $c$ results, the FNR is the same as the False Positive Rate, the fraction of incorrectly selected queries.

The FNR metric has some limitations.  First, missing the highest query will be penalized the same as missing the $c$-th one.  Second, selecting a query with a very low score will be penalized the same as selecting the $(c+1)$-th query, whose score may be quite close to the $c$'th query.  We thus use another metric that we call Score Error Rate (SER), which measures the ratio of ``missed scores'' by selecting $S$ instead of the true top $c$ queries, denoted by $\mathsf{Top}_c$.
	$$\mathit{SER} = 1.0 - \frac{\mathrm{avgScore}(S)}{\mathrm{avgScore}(\mathsf{Top}_c)}.$$

We present results for both FNR and SER and observe that the correlation between them is quite stable.

\begin{table}
\centering
\begin{tabular}{|c|c|c|c|c|c|}\hline
{\bf Dataset} & {\bf Number of Records} & {\bf Number of Items}\\ \hline
BMS-POS & 515,597 & 1,657\\ \hline
Kosarak & 990,002 & 41,270\\ \hline
AOL & 647,377 & 2,290,685\\ \hline
Zipf & 1,000,000 & 10,000\\ \hline
\end{tabular}
\caption{Dataset characteristics}\label{tab:dataset_desc}
\end{table}

\begin{table}
\centering
\begin{tabular}{|c|c|c|}\hline
{\bf Settings} & {\bf Methods} & {\bf Description}\\ \hline
\multirow{2}{*}{Interactive} & SVT-DPBook & DPBook SVT (Alg.~\ref{alg:sparsebook}).  \\ \cline{2-3}
& SVT-S & Standard SVT (Alg.~\ref{alg:standard}). \\ \hline
\multirow{2}{*}{Non-interactive} & SVT-ReTr & Standard SVT with Retraversal. \\ \cline{2-3}
& EM & Exponential Mechanism. \\ \hline
\end{tabular}
\caption{Summary of algorithms}\label{tab:methods_summary}
\end{table}

\begin{figure}[h]
	\includegraphics[width = 3.0in]{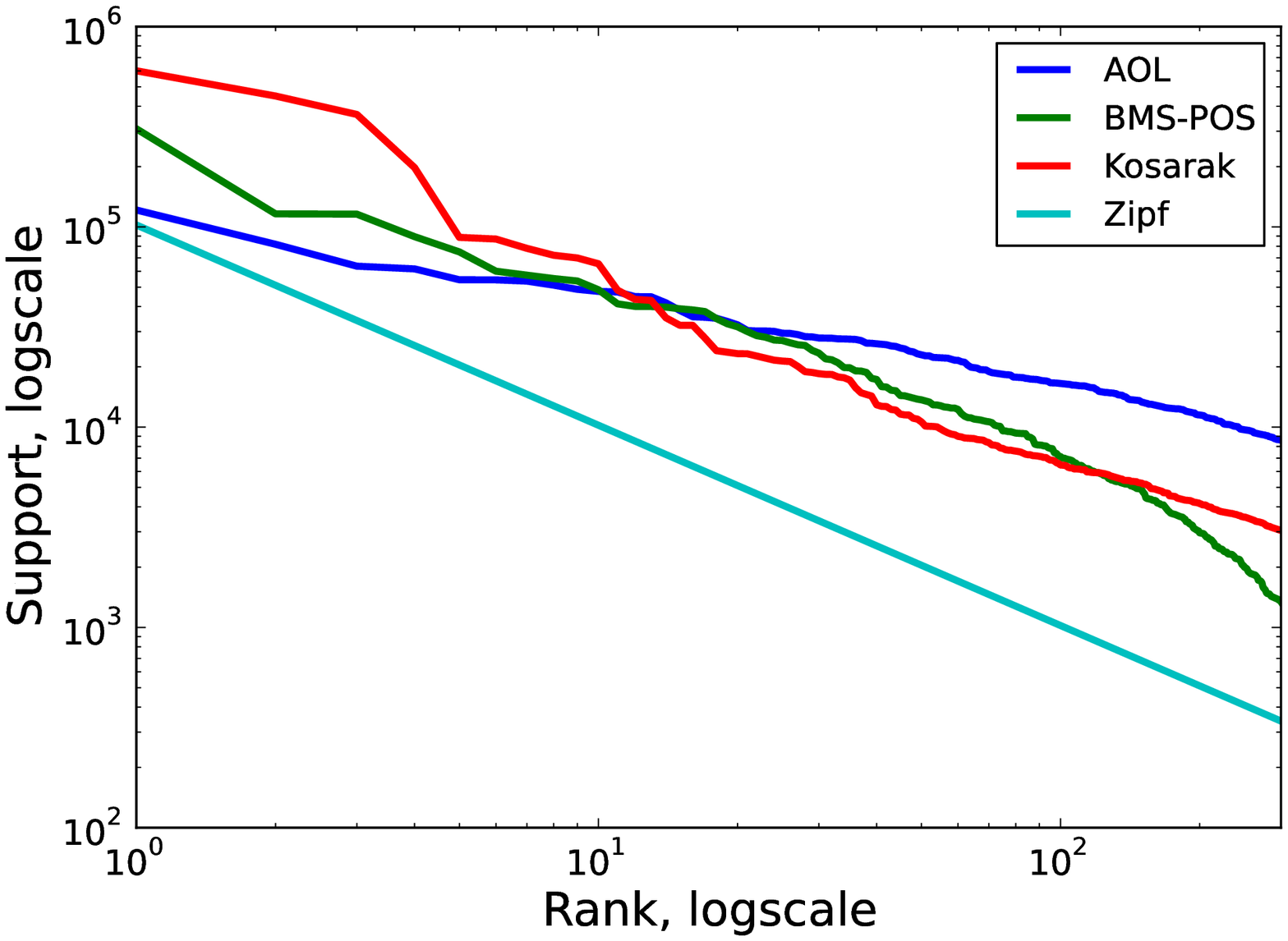}
	\caption{The distribution of 300 highest scores from experiment datasets.}\label{fig:zipf_plot}
\end{figure}

\mypara{Datasets.}  The performance of different algorithms would be affected by the distribution of query scores, we thus want to evaluate the algorithms on several representative distributions.  In the experiments, we use the item frequencies in three real datasets: BMS-POS, Kosarak and AOL as representative distributions of query scores.
In addition, we also use the distribution inspired by the Zipf's law, which states that given some corpus of natural language utterances, the frequency of any word is inversely proportional to its rank in the frequency table.  Similar phenomenon occurs in many other rankings unrelated to language, such as the population ranks of cities in various countries, corporation sizes, income rankings, ranks of number of people watching the same TV channel, and so on.
In this distribution, the $i$'th query has a score proportional to $\frac{1}{i}$.
The characteristics of these datasets are summarized in Table~\ref{tab:dataset_desc}, and the distribution of the $300$ highest scores are shown in Figure~\ref{fig:zipf_plot}.

\begin{figure*}[p]
\begin{tabular}{cc}
    \multicolumn{2}{c}{\includegraphics[width = 6.0in]{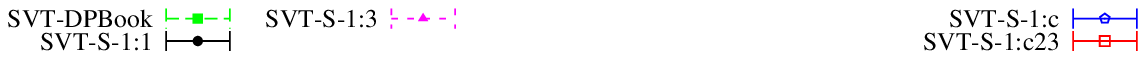}}\\
	\includegraphics[width = 3.0in]{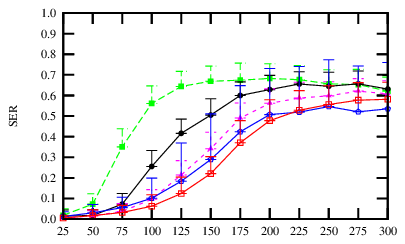}
	& \includegraphics[width = 3.0in]{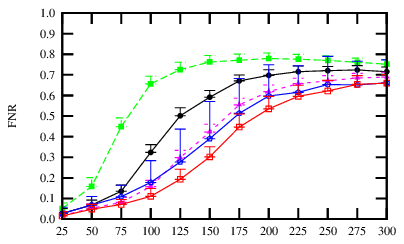}\\
	(a) BMS-POS, SER & (b) BMS-POS, FNR\\
	\includegraphics[width = 3.0in]{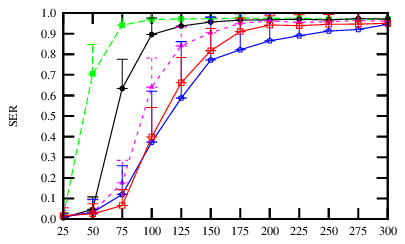}
	& \includegraphics[width = 3.0in]{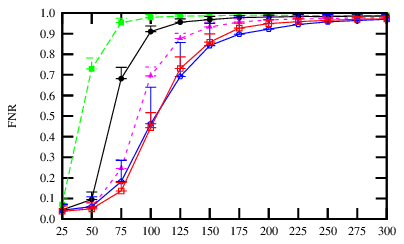}\\
	(c) Kosarak, SER & (d) Kosarak, FNR\\
	\includegraphics[width = 3.0in]{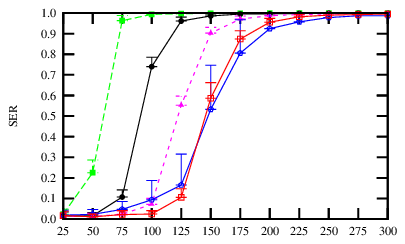}
	& \includegraphics[width = 3.0in]{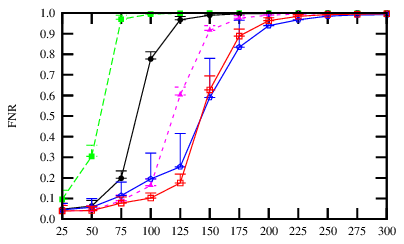}\\
	(e) AOL, SER & (f) AOL, FNR\\
	\includegraphics[width = 3.0in]{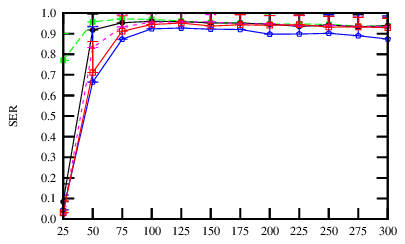}
	& \includegraphics[width = 3.0in]{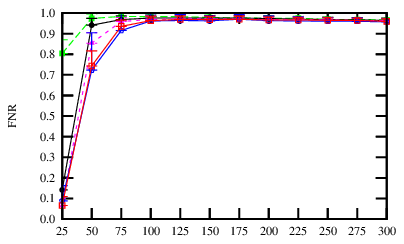}\\
	(g) Zipf Synthe, SER & (h) Zipf Synthe, FNR\\

\end{tabular}
	\caption{Comparison of interactive approaches: SVT-DPBook and SVT-S with different budget allocation.  Privacy budget $\epsilon = 0.1$.  x-axis: top-$c$}\label{fig:comp_Interactive}
\end{figure*}

\begin{figure*}[p]
\begin{tabular}{cc}
    \multicolumn{2}{c}{\includegraphics[width = 6.0in]{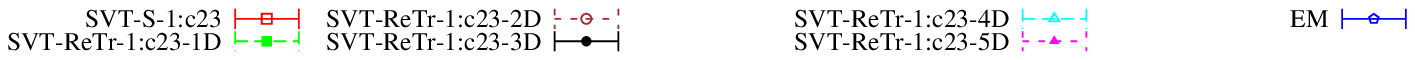}}\\
	\includegraphics[width = 3.0in]{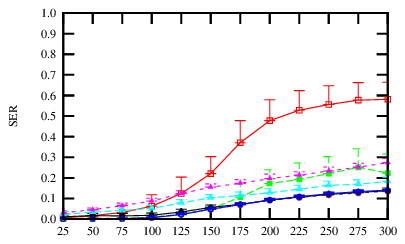}
	& \includegraphics[width = 3.0in]{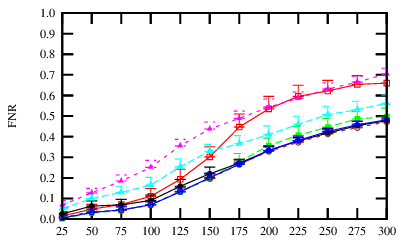}\\
	(a) BMS-POS, SER & (b) BMS-POS, FNR\\
	\includegraphics[width = 3.0in]{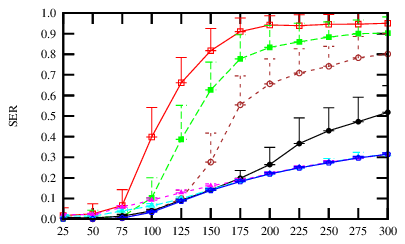}
	& \includegraphics[width = 3.0in]{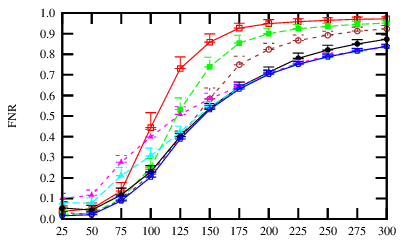}\\
	(c) Kosarak, SER & (d) Kosarak, FNR\\
	\includegraphics[width = 3.0in]{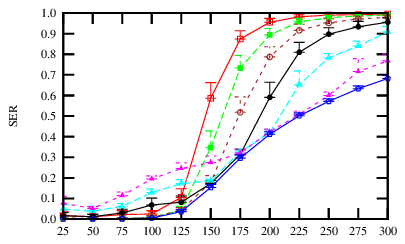}
	& \includegraphics[width = 3.0in]{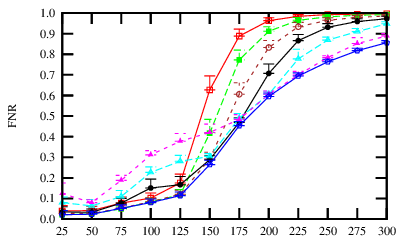}\\
	(e) AOL, SER & (f) AOL, FNR\\
	\includegraphics[width = 3.0in]{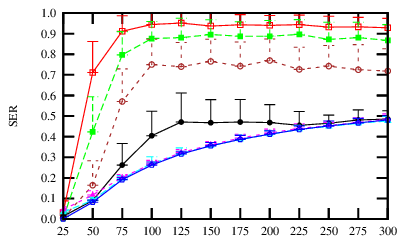}
	& \includegraphics[width = 3.0in]{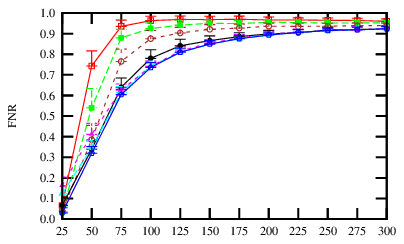}\\
	(g) Zipf Synthe, SER & (h) Zipf Synthe, FNR\\

\end{tabular}
	\caption{Comparison of non-interactive approaches: EM and SVT-ReTr with different thresholds.  Privacy budget $\epsilon = 0.1$.  x-axis: top-$c$.}\label{fig:comp_noninteractive}
\end{figure*}

\mypara{Evaluation Setup.}  We consider the following algorithms.  SVT-DPBook is from Dwork and Roth's 2014 book~\cite{DPBook} (Alg.~\ref{alg:sparsebook}).  SVT-S is our proposed standard SVT, i.e., Alg.~\ref{alg:standard} without numerical outputs ($\epsilon_3 = 0$); and since the count query is monotonic, we use the version for monotonic queries in Section~\ref{sec:mono_queries}. We consider four privacy budget allocations, 1:1, 1:3, 1:$c$ and 1:$c^{2/3}$, where the last is what our analysis suggests for the monotonic case.  These algorithms can be applied in both the interactive and the non-interactive setting.

For the non-interactive setting, we consider EM and SVT-ReTr, which is SVT with the optimizations of increasing the threshold and retraversing through the queries (items) until $c$ of them are selected.  We fix the privacy budget allocation to be $1:c^{2/3}$ and vary the amount we increase the threshold from 1D, 2D, $\ldots$, to 5D, where 1D means adding one standard deviation of the added noises to the threshold.

We vary $c$ from 25 to 300, and each time uses the average score for the $c$'th query and the $c+1$'th query as the threshold.  We show results for privacy budget $\epsilon=0.1$ in the paper.  We omit results for other $\epsilon$ values because of space limitation.  We note that varying $c$ have a similar impact of varying $\epsilon$, since the accuracy of each method is mostly affect by $\frac{\epsilon}{c}$; therefore the impact of different $\epsilon$ can be seen from different $c$ values.  
We run each experiment 100 times, each time randomizing the order of items to be examined.  We report the average and standard deviation of SER.  All algorithms are implemented in Python 2.7 and all the experiments are conducted on an Intel Core i7-3770 3.40GHz PC with 16GB memory.

\mypara{Results in the Interactive Setting.}
Figure~\ref{fig:comp_Interactive} reports the results for the algorithms that can be applied in the interactive setting.
While it is clear that in some settings (such as when $c=25$) all methods are quite accurate, and in some other settings all methods are very inaccurate (such as when $c \geq 100$ for the Zipf dataset), in the settings in between the two extremes, the differences between these methods are quite large.

SVT-DPBook performs the worst, followed by SVT-S-1:1, then by SVT-S-1:3, and finally by SVT-S-1:c and SVT-S-1:c23.  The differences among these algorithms can be quite pronounced.  For example, on the Kosarak dataset, with $\epsilon=0.1$, $c=50$, SVT-DPBook's SER is $0.705$, which means that the average support of selected items is only around $30\%$ of that for the true top-$50$ items, which we interpret to mean that the output is meaningless.  In contrast, all four SVT-S algorithms have SER less than 0.05, suggesting high accuracy in the selection.
SVT-DPBook's poor performance is due to the fact that the threshold is perturbed by a noise with scale as large as $c\Delta/\epsilon$.

For the differences among the four budget allocation approaches, it appears that the performance of $1:c$ and $1:c^{\frac{2}{3}}$ are clearly better than the others; and their advantages over the standard $1:1$ allocation is quite pronounced.  Which of $1:c$ and $1:c^{\frac{2}{3}}$ is better is less clear.  In general, the former is better for larger $c$ values, where the error is higher, and the latter is better for smaller $c$ values, where the error is lower.  Also note that $1:c$ results in a significantly larger standard deviation.  For these reasons, we interpret the results as supporting the recommendation of using  1:$c^{2/3}$ budget allocation.



\mypara{Results in the Non-interactive Setting.}
Figure~\ref{fig:comp_noninteractive} reports the results for the algorithms that can be applied in the noninteractive setting.  We observe that EM clearly performs better than
SVT-ReTr-1:c23, which performs essentially the same as SVT-S-1:c23, which is the best algorithm for the interactive case, and is already much better than SVT algorithms used in the literature.  For example, for the AOL dataset with $c=150$, EM's SER is $0.15$, while SVT-S with $1:c^{2/3}$ allocation has SER of $0.59$, and SVT-S with $1:1$ allocation has SER of $0.99$.

It is interesting to see that increasing the threshold can significantly improve the accuracy of SVT with Retraversal.  However, the best threshold increment value depends on the dataset and the number of items to be selected.  For example, 5D works well for Zipf, and for Kosarak and AOL when $c$ is large, but works not as well for BMS and for Kosarak and AOL when $c$ is small.
Since it is unclear how to select the best threshold increment value, and even with the best threshold increment, SVT-ReTr performs no better than EM, our experiments suggest that usage of SVT should be replaced by EM in the non-interactive setting.

\section{Related Work} \label{sec:related}
SVT was introduced by  Dwork et al.~\cite{DNR+09}, and improved by Roth and Roughgarden~\cite{RR10} and by Hardt and Rothblum~\cite{HR10}.  
These usages are in an interactive setting.  An early description of SVT as a stand-alone technique appeared in Roth's 2011 lecture notes~\cite{Roth2011}, which is Alg.~\ref{alg:sparselecture} in this paper, and is in fact $\infty$-DP.  The algorithms in~\cite{RR10,HR10} also has another difference, as discussed in Section~\ref{sec:other}. Another version of SVT appeared in the 2014 book~\cite{DPBook}, which is Alg.~\ref{alg:sparsebook}.  This version is used in some papers, e.g.,~\cite{SS15}.  We show that it is possible to add less noise and obtain higher accuracy for the same privacy parameter.


Lee and Clifton \cite{LC14} used a variant of SVT  (see Algorithm \ref{alg:LC}) to find itemsets whose support is above  the threshold. Stoddard et al.~\cite{SCM14} proposed another variant (see Algorithm \ref{alg:SCM}) for private feature selection for classification to pick out  the set of features with scores greater than the perturbed threshold.
Chen et al.~\cite{CXZX15}  employed yet another variant of  SVT (see Algorithm \ref{alg:CXZX}) to return attribute pairs with mutual information  greater than the corresponding noisy threshold.
These usages are not private.
Some of these errors were pointed in~\cite{CM15}, in which a generalized private threshold testing algorithm (GPTT) that attempts to model the SVT variants in \cite{LC14, SCM14, CXZX15} was introduced. The authors showed that GPTT did not satisfy $\epsilon'$-DP for any finite $\epsilon'$. But there is an error in the proof, as shown in Section~\ref{sec:gptt}. Independent from our work, Zhang et al. \cite{PrivTree} presented two proofs that the variant of SVT violates DP without discussing the cause of the errors.  Also presented in~\cite{PrivTree} is a special case of our proposed Alg.~1 for counting queries.
To our knowledge, the general version of our improved SVT (Alg.~\ref{alg:sparseour} and Alg.~\ref{alg:standard}), the techniques of optimizing budget allocation, the technique of using re-traversal to improve SVT, and the comparison of SVT and EM are new in our work.


\section{Conclusion} \label{sec:conclusion}
We have introduced a new version of SVT that provides better utility.  We also introduce an effective technique to improve the performance of SVT by optimizing the distribution of privacy budget.  These enhancements achieve better utility than the state of the art SVT and can be applied to improve utility in the interactive setting.  We have also explained the misunderstandings and errors in a number of papers that use or analyze SVT; and believe that these will help clarify the misunderstandings regarding SVT and help avoid similar errors in the future.  We have also shown that in the non-interactive setting, EM should be preferred over SVT.



\bibliographystyle{abbrv}
{\bibliography{main}}

\section{Appendix}\label{app:proofs}

\subsection{Proof that Alg.~\ref{alg:sparselecture} is non-private}

\begin{theorem}
\label{prop:Rothinfinite}
Alg.~\ref{alg:sparselecture} is not $\epsilon^{\prime}$-DP for any finite $\epsilon^{\prime}$.
\end{theorem}
\begin{proof}
Set $c=1$ for simplicity.  
Given any finite $\epsilon^{\prime}>0$, we construct an example to show that Alg.~\ref{alg:sparselecture} is not $\epsilon^{\prime}$-DP.  
Consider an example with $T=0$, and $m+1$ queries $\mathbf{q}$ with sensitivity $\Delta$ such that  $\mathbf{q}(D)=0^m\Delta$ and $\mathbf{q}(D^{\prime})=\Delta^m 0$,
and the output vector $\vec{a}=\nega^m 0$, that is, only the last query answer is a numeric value $0$.  Let $\AA$ be Alg.~\ref{alg:sparselecture}.  We show that  $\frac{\Pr{\AA(D)=\scriptsize{\vec{a}}}}{\Pr{\AA(D^{\prime})=\scriptsize{\vec{a}}}} \ge e^{\epsilon^{\prime}} $ for any $ \epsilon^{\prime}>0$ when $m$  is large enough.

We denote the cumulative distribution function of $\Lap{\frac{2\Delta}{\epsilon}}$ by $F(x)$. We have
{\small
\begin{align}
 \nonumber
 & \Pr{\AA(D)=\vec{a}}  \\
 \nonumber
 =& \int_{-\infty}^{\infty} \!\Pr{\rho\!=\!z} f_{D}(z)\Pr{\Delta\splus \nu_{m+1}\!\ge\!  z  \wedge \Delta\splus \nu_{m+1} \!=\!0}  \: dz  \\
 \nonumber
 =& \int_{-\infty}^{\infty}\Pr{\rho\!=\!z}f_{D}(z) \Pr{0\ge z} \Pr{\nu_{m+1} = -\Delta}  \: dz \\
 \nonumber
 =& \frac{\epsilon}{4\Delta}e^{-\frac{\epsilon}{2}} \int_{-\infty}^{\infty}\Pr{\rho\!=\!z}f_{D}(z) \Pr{0\ge z}  \: dz \\
 \nonumber
 = & \frac{\epsilon}{4\Delta}e^{-\frac{\epsilon}{2}} \int_{-\infty}^{0} \Pr{\rho\!=\!z}f_{D}(z) \: dz    \\
 \nonumber
= & \frac{\epsilon}{4\Delta} e^{-\frac{\epsilon}{2}}\int_{-\infty}^0  \Pr{\rho=z} \prod\limits_{i=1}^{m} \Pr{\nu_i< z}  \: dz\\
 \label{eq:three_proof1}
 = &  \frac{\epsilon}{4\Delta} e^{-\frac{\epsilon}{2}}\int_{-\infty}^0  \Pr{\rho=z}(F(z))^m \: dz, \\
 \nonumber
 & \mbox{  and similarly} \\
 \label{eq:three_proof2}
 & \Pr{\AA(D^{\prime})=\vec{a}} = \frac{\epsilon}{4\Delta} \int_{-\infty}^0 \Pr{\rho=z^{\prime}} (F(z^{\prime}-\Delta))^{m}  \: dz^{\prime}.
\end{align}
}

The fact that $0$ is given as an output reveals the information that the noisy threshold is at most $0$, forcing the range of integration to be from $-\infty$ to $0$, instead of from $-\infty$ to $\infty$.  This prevents the use of changing $z$ in (\ref{eq:three_proof1}) to $z^{\prime}-\Delta$ to bound the ratio of (\ref{eq:three_proof1}) to (\ref{eq:three_proof2}).

Noting that $\frac{ F(z)}{F(z-\Delta)}=e^{\frac{\epsilon}{2}}$ for any $z \le 0$, we thus have
{\small
\begin{align*}
 \frac{\Pr{\AA(D)=\vec{a}}}{\Pr{\AA(D^{\prime})=\vec{a}}} &= e^{-\frac{\epsilon}{2}}\frac{ \int_{-\infty}^0  \Pr{\rho=z}(F(z))^m \: dz }{ \int_{-\infty}^0 \Pr{\rho=z^{\prime}} (F(z^{\prime}-\Delta))^{m}  \: dz^{\prime}} \\
&=e^{-\frac{\epsilon}{2}}\frac{ \int_{-\infty}^0  \Pr{\rho=z}( e^{\frac{\epsilon}{2}}F(z-\Delta))^m \: dz }{ \int_{-\infty}^0 \Pr{\rho=z^{\prime}} (F(z^{\prime}-\Delta))^{m}  \: dz^{\prime}} \\
&=e^{(m-1)\frac{\epsilon}{2}},
\end{align*}
}
and thus when $m>\lceil \frac{2\epsilon^{\prime}}{\epsilon}\rceil+1$, we have  $\frac{\Pr{\AA(D)=\scriptsize{\vec{a}}}}{\Pr{\AA(D^{\prime})=\scriptsize{\vec{a}}}} > e^{\epsilon^{\prime}}$.
\end{proof}

\subsection{Proof that Alg.~\ref{alg:CXZX} is non-private}

\begin{theorem}
\label{prop:Alg6infinite}
Alg.~\ref{alg:CXZX} is not $\epsilon^{\prime}$-DP for any finite $\epsilon^{\prime}$.
\end{theorem}
 \begin{proof}
We construct a counterexample with $\Delta=1$, $T=0$, and $2m$ queries such that $\mathbf{q}(D)=0^{2m}$, and   $\mathbf{q}(D^{\prime})=1^m(-1)^m$. Consider the output vector $\vec{a}=\nega^m\posi^m$. Denote the cumulative distribution function of $\nu_i$ by $F(x)$. 
From Eq.~(\ref{eq:prgen}), we have
{\small
\begin{align*}
& \Pr{\AA(D)=\vec{a}}  \\
&=  \int_{-\infty}^{\infty}  \Pr{\rho=z}\prod\limits_{i=1}^{m} \Pr{0\splus \nu_i< z} \prod\limits_{i=m+1}^{2m} \Pr{0\splus \nu_i\ge z}\: dz\\
&= \int_{-\infty}^{\infty} \Pr{\rho=z} ( F(z) (1-F(z)) ) ^m \: dz,
\end{align*}
}
and
{\small
\begin{align*}
& \Pr{\AA(D')=\vec{a}}  \\
&=  \int_{-\infty}^{\infty}  \Pr{\rho=z}\prod\limits_{i=1}^{m} \Pr{1\splus \nu_i< z} \prod\limits_{i=m+1}^{2m} \Pr{-1\splus \nu_i\ge z}\: dz\\
&= \int_{-\infty}^{\infty} \Pr{\rho=z} ( F(z-1) (1-F(z\splus 1)) ) ^m \: dz.
\end{align*}
}
We now show that  $\frac{\Pr{\AA(D)=\scriptsize{\vec{a}}} }{\Pr{\AA(D')=\scriptsize{\vec{a}}} }$  is unbounded as $m$ increases, proving this theorem. Compare  $F(z) (1-F(z))$ with $F(z-1) (1-F(z\splus 1))$. Note that $F(z)$ is  monotonically  increasing. When $  z \le 0$,

\vspace*{-0.2in}{\small
\begin{equation*}
\frac{F(z) (1-F(z)) }{F(z-1) (1-F(z\splus 1)) } \ge \frac{F(z) }{F(z-1)}
=\frac{\frac{1}{2}e^{\frac{\epsilon}{2}z} }{\frac{1}{2}e^{\frac{\epsilon}{2}(z-1)}} =e^{\frac{\epsilon}{2}}  .
\end{equation*}
}
When $z >0$, we also have
{\small
\begin{equation*}
\frac{F(z) (1-F(z)) }{F(z-1) (1-F(z\splus 1)) } \ge \frac{1-F(z) }{1-F(z+1)}
=\frac{\frac{1}{2}e^{-\frac{\epsilon}{2}z} }{\frac{1}{2}e^{-\frac{\epsilon}{2}(z+1)}} =e^{\frac{\epsilon}{2}}  .
\end{equation*}
}
So,  $\frac{\Pr{\AA(D)=\scriptsize{\vec{a}}} }{\Pr{\AA(D')=\scriptsize{\vec{a}}} } \ge e^{\frac{m\epsilon}{2}}$, which is greater than $e^{\epsilon'}$ when $m > \lceil \frac{2\epsilon'}{\epsilon}\rceil$ for any finite $\epsilon'$.
\end{proof}

\subsection{Error of non-privacy proof in~\protect\cite{CM15}}\label{app:gptt}

The proof in~\cite{CM15} that GPTT is non-private considers the counter-example  with $\Delta=1$, $T=0$,  a sequence $\mathbf{q}$ of $2t$ queries such that $\mathbf{q}(D)=0^t1^t$ and $\mathbf{q}(D')=1^t0^t$, and the output vector $\vec{a}=\nega^t\posi^t$.  Then
{\scriptsize
\begin{align*}
\frac{\Pr{\GPTT(D) =\vec{a}}}{\Pr{\GPTT(D') =\vec{a}}}
&= \frac{\int_{
\sminus\infty}^{\infty}  \Pr{\rho=z} \left(F_{\epsilon_2}(z)\sminus F_{\epsilon_2}(z)F_{\epsilon_2}(z\sminus 1)\right)^t dz}{\int_{-\infty}^{\infty}  \Pr{\rho=z} \left(F_{\epsilon_2}(z\sminus1)\sminus F_{\epsilon_2}(z)F_{\epsilon_2}(z\sminus 1)\right)^t dz}
       \\
{\small \mbox{where}\;\; F_{\epsilon}(x) }& {\small \mbox{  is the cumulative distribution function of } \Lap{1/\epsilon}.}
\end{align*}
}
The goal of the proof is to show that the above is unbounded as $t$ increases.  A key observation is that the ratio of
the integrands of the two integrals is always larger than $1$, i.e.,
\begin{equation*}
\kappa(z)=\frac{F_{\epsilon_2}(z)-F_{\epsilon_2}(z)F_{\epsilon_2}(z-1)}{F_{\epsilon_2}(z-1)-F_{\epsilon_2}(z)F_{\epsilon_2}(z-1)} > 1
\end{equation*}
For example, since $F_{\epsilon}(x)$ is the cumulative distribution function of $\Lap{1/\epsilon}$, we have $F_{\epsilon_2}(0)=1/2$ and $F_{\epsilon_2}(-1)<1/2$; and thus $\kappa(0)=\frac{1-F_{\epsilon_2}(-1)}{F_{\epsilon_2}(-1)}>1$.
However, when $|z|$ goes to $\infty$, $\kappa(z)$ goes to $1$.
Thus the proof tries to limit the integrals to be a finite interval so that there is a lower-bound for $\kappa(z)$ that is greater than $1$.  It denotes $\alpha=\Pr{\GPTT(D') =\vec{a}}$. Then choose parameter $\delta=|F_{\epsilon_1}^{-1}(\frac{\alpha}{4})|$ to use $[-\delta,\delta]$ as the finite interval, and thus
\begin{align*}
\alpha \leq 2 \int_{-\delta}^{\delta}  \Pr{\rho=z} \left(F_{\epsilon_2}(z-1)-F_{\epsilon_2}(z)F_{\epsilon_2}(z-1)\right)^t dz.
\end{align*}

Denote the minimum of $\kappa(z)$ in the closed interval $[-\delta, \delta]$ by $\kappa$. Then we have
$\frac{\Pr{\GPTT(D)=\vec{a}}}{\Pr{\GPTT(D')=\vec{a}}}> \frac{\kappa^t}{2} $.
The proof claims that  for any $\epsilon'>1$ there exists a $t$ to make the above ratio larger than $e^{\epsilon'}$.

The proof is incorrect because of dependency in the parameters.  First, $\alpha$ is a function of $t$; and when $t$ increases, $\alpha$ decreases because the integrand above is positive and decreasing.  Second, $\delta$ depends on $\alpha$, and when $\alpha$ decreases, $\delta$ increases.  Thus when $t$ increases, $\delta$ increases. We write $\delta$ as $\delta(t)$ to make the dependency on $t$ explicit.  Third, $\kappa$, the minimum value of $\kappa(z)$ over the interval $[-\delta(t),\delta(t)]$, decreases when $t$ increases. That is, $\kappa$ is also dependent on $t$, denoted by $\kappa(t)$, and  decreases while $t$ increases. It is not sure that there exists such a $t$  that $\frac{\kappa(t)^t}{2}> e^{\epsilon'}$ for any $\epsilon'>1$.

To demonstrate that the error in the proof cannot be easily fixed, we point out that following the logic of that proof, one can prove that Alg.~\ref{alg:sparseour} is not $\epsilon'$-DP for any finite $\epsilon'$.  We now show such a ``proof'' that contradicts Lemma~\ref{lemma:negative}.  Let $\AA$ be Alg.~\ref{alg:sparseour}, with $c=1$.  Consider an example with $\Delta=1$, $T=0$, a sequence $\mathbf{q}$ of $t$ queries such that $\mathbf{q}(D)=0^t$ and $\mathbf{q}(D')=1^t$, and output vector $\vec{a}=\nega^t$.  Let
{\small
\begin{align*}
& \beta = \Pr{\AA(D) =\nega^\ell}   =  \int_{-\infty}^{\infty}  \Pr{\rho=z} \left(F_{\frac{\epsilon}{4}}(z)\right)^t dz
 \\
& \alpha  =\Pr{\AA(D') =\nega^\ell}   =  \int_{-\infty}^{\infty}  \Pr{\rho=z} \left(F_{\frac{\epsilon}{4}}(z-1)\right)^t dz,
 \\
&\quad \mbox {where }\;\;  F_{\frac{\epsilon}{4}}(x) \mbox{  is the cumulative distribution function of } \Lap{\frac{4}{\epsilon}}.
\end{align*}
}
Find a parameter $\delta$ such that $\int_{-\delta}^{\delta} \Pr{\rho=z}\: dz\geq 1-\frac{\alpha}{2}$. Then $\int_{-\delta}^{\delta}  \Pr{\rho=z} \left(F_{\frac{\epsilon}{4}}(z-1)\right)^t dz \geq \frac{\alpha}{2}$.  Let $\kappa$ be the minimum value of $\frac{F_{\frac{\epsilon}{4}}(z)}{F_{\frac{\epsilon}{4}}(z-1)}$ in $[-\delta,\delta]$; it must be that $\kappa>1$.  Then
{\scriptsize
\begin{align*}
\beta & > \int_{-\delta}^{\delta}  \Pr{\rho=z} \left(F_{\frac{\epsilon}{4}}(z)\right)^t dz \geq \int_{-\delta}^{\delta}  \Pr{\rho=z} \left(\kappa F_{\frac{\epsilon}{4}}(z-1)\right)^t dz \\
 & = \kappa^t \int_{-\delta}^{\delta}  \Pr{\rho=z} \left(F_{\frac{\epsilon}{4}}(z-1)\right)^t dz \geq \frac{\kappa^t}{2}\alpha.
\end{align*}
}
Since $\kappa>1$, one can choose a large enough $t$ to make $\frac{\beta}{\alpha}=\frac{\kappa^t}{2}$ to be as large as needed.   We note that this contradicts Lemma~\ref{lemma:negative}. The contradiction shows that the proof logic used in~\cite{CM15} is incorrect.


\end{document}